\documentclass[a4paper,11pt,fleqn]{article}

\usepackage[ansinew]{inputenc}
\usepackage[mathscr]{eucal}
\usepackage{amsmath,amssymb,amsthm,cite}
\usepackage{graphicx}
\usepackage{color}

\allowdisplaybreaks

\newcommand{\p}{\partial}

\newcommand{\todo}[1][\null]{\ensuremath{\clubsuit}}

\newcommand{\noprint}[1]{}

\newtheorem{theorem}{Theorem}

\newtheorem{corollary}{Corollary}
\newtheorem{proposition}{Proposition}
{
\theoremstyle{definition}

\newtheorem*{remark*}{Remark}
}

\newcommand{\lsemioplus}{\mathbin{\mbox{$\lefteqn{\hspace{.77ex}\rule{.4pt}{1.2ex}}{\in}$}}}

\newcommand{\checked}[1][\null]{\ensuremath{\boldsymbol{\surd}}}

\newcommand{\PP}{\mathcal{P}}
\newcommand{\ZZ}{\mathcal{Z}}

\newcommand{\DDD}{\mathcal{D}}
\newcommand{\JJJ}{\mathcal{J}}

\newcommand{\sbve}{sBVE}

\begin{document}

\par\noindent {\LARGE\bf
Complete point symmetry group of\\ the barotropic vorticity equation on a rotating sphere
\par}

{\vspace{4mm}\par\noindent {\bf Elsa Dos Santos Cardoso-Bihlo$^\dag$ and Roman O.\ Popovych$^{\dag\, \ddag}$
} \par\vspace{2mm}\par}

{\vspace{2mm}\par\noindent {\it
$^{\dag}$~Wolfgang Pauli Institute, Nordbergstra{\ss}e 15, A-1090 Vienna, Austria
}}

{\vspace{2mm}\par\noindent {\it
$^{\ddag}$~Institute of Mathematics of NAS of Ukraine, 3 Tereshchenkivska Str., 01601 Kyiv, Ukraine
}\par}

{\vspace{2mm}\par\noindent 
$\phantom{^\dag}$~\textup{E-mail}:
$^{\dag}$elsa.cardoso@univie.ac.at,
$^{\ddag}$rop@imath.kiev.ua
\par}

\vspace{4mm}\par\noindent\hspace*{5mm}\parbox{150mm}{\small
The complete point symmetry group of the barotropic vorticity equation on the sphere is determined.
The method we use relies on the invariance of megaideals of the maximal Lie invariance algebra of a system of differential equations under automorphisms generated by the associated point symmetry group.
A convenient set of megaideals is found for the maximal Lie invariance algebra of the spherical vorticity equation.
We prove that there are only two independent (up to composition with continuous point symmetry transformations) discrete symmetries for this equation.
}\par\vspace{2mm}

\section{Introduction}

\looseness=-1
Lie symmetries of the inviscid barotropic vorticity equation on the rotating sphere (\sbve) were computed in~\cite{bihl09Ay} and used in~\cite{bihl09Ay,bihl12Ay}
in order to derive a point transformation mapping the equation in a rotating reference frame to the equation in the rest frame and to construct exact solutions.
Therein it was also indicated that besides an infinite-dimensional maximal Lie invariance group
the barotropic vorticity equation on the sphere admits two independent (up to composition with each other and with continuous transformations) discrete symmetries,
which merely alternate signs of two pairs of variables, (time, longitude) and (latitude, stream function), respectively.
However, no systematic derivation of discrete symmetries or, more generally, the complete point symmetry group of the \sbve\ has been given in the literature up to now.

With the present paper, we aim to complete the description of point symmetries admitted by the \sbve\ by computing the complete point symmetry group~$G_\Omega$ of this equation.
We simplify the computation within the framework of the direct method by combining it with an advanced version of the algebraic approach originally proposed in~\cite{hydo00By,hydo00Ay}, essentially modified in~\cite{bihl11Cy} and then developed in~\cite{bihl11Dy,bihl11By}.
As a result, we prove that in fact the group~$G_\Omega$ is generated by Lie symmetry transformations of the \sbve\ and the previously mentioned two discrete transformations.

\looseness=-1
The \sbve\ is an appropriate equation to demonstrate the advantages of the enhanced version of the algebraic approach.
In principle, the group~$G_\Omega$ might be computed using merely the direct method based on
the definition of the set of transformations to be found and the prolongation of finite transformations to derivatives by the chain rule.
This approach is widely applied in the literature for finding complete point symmetry groups of single systems of differential equations
or equivalence groups and equivalence groupoids (i.e., sets of all admissible point transformations) of classes of such systems, 
see, e.g., \cite{king91Ay,king98Ay,king01Ay,popo12Ay,popo10Ay,vane07Ay,vane09Ay} and references therein.
At the same time, the \sbve\ is a third-order nonlinear partial differential equation for a single scalar function of three independent variables,
and one of the independent variables explicitly appears in the equation.
Hence the mere application of the direct method for the computation of the group~$G_\Omega$ is too cumbersome.
Moreover, the \sbve\ admits an infinite-dimensional maximal Lie invariance algebra,
and therefore the version of the algebraic method proposed in~\cite{hydo00By,hydo00Ay} is not applicable here,
because it relies on the computation of automorphism matrices,
which properly works only in the finite-dimensional case.
From the physical point of view, the \sbve\ is of superior importance
as it is capable of describing qualitatively the large-scale behavior of the flow in the middle of the troposphere.
Due to its relevance for the larger atmospheric scales,
it is especially convenient to consider the vorticity equation in spherical coordinates.

Our paper is organized in the following way:
In Section~\ref{sec:TheModel} we recall the known point symmetries of the \sbve.
In Section~\ref{sec:DiscreteSymmetries} we describe a method that can be applied to determine the complete point symmetry group of a system of differential equations
possessing a nontrivial Lie invariance algebra.
This method involves the notion of megaideals of Lie algebras and properly works even for systems whose maximal Lie invariance algebras are infinite dimensional.
Section~\ref{sec:CompletePointSymmetryGroup} is central.
After determining a convenient set of megaideals for the maximal Lie invariance algebra of the vorticity equation on the sphere,
we derive a maximal system of constraints for elements of the group~$G_\Omega$, which are related to properties of the adjoint action of~$G_\Omega$ on its Lie algebra,
and then complete the computation of~$G_\Omega$ by the direct method using the constraints derived within the framework of the algebraic approach.
We briefly sum up our results in the conclusion.

\section{The model}\label{sec:TheModel}

Introducing the stream function in the system of nonlinear incompressible Euler equations in a single thin atmospheric layer on the sphere 
leads to a single third-order nonlinear partial differential equation for the stream function, which is referred to as the barotropic vorticity equation on the sphere. It reads
\begin{equation}\label{eq:VorticityEquationSphere}
 \zeta_t + \psi_\lambda\zeta_\mu - \psi_\mu\zeta_\lambda + 2\Omega\psi_\lambda = 0,\qquad
 \zeta := \frac{1}{1-\mu^2}\,\psi_{\lambda\lambda}+((1-\mu^2)\psi_\mu)_\mu,
\end{equation}
where $\lambda$ and $\varphi$ are the longitude and latitude, respectively, and $\mu=\sin\varphi$, $\psi$ is the stream function generating an incompressible two-dimensional flow on the sphere, which is related to the vorticity~$\zeta$ by means of the Laplacian on the sphere and $\Omega$ is the constant angular velocity of the rotating sphere. The derived latitudinal variable $\mu$ runs from $-1$ (South Pole) to $1$ (North Pole). For convenience, we assume the mean radius of the Earth to be scaled to one.

It was shown in~\cite{bihl09Ay,bihl12Ay} that Eq.~\eqref{eq:VorticityEquationSphere} admits the infinite-dimensional maximal Lie invariance algebra, which is denoted as $\mathcal S^\infty_\Omega$. This algebra is generated by the vector fields
\begin{align*}
     &\DDD = t\p_t - (\psi-\Omega\mu)\p_{\psi} - \Omega t\p_{\lambda}, \qquad \PP=\p_t, \qquad \ZZ(g) = g(t)\p_{\psi},\qquad \JJJ_{1} = \p_{\lambda}, \\
     &\JJJ_{2} = \mu\frac{\sin(\lambda+\Omega t)}{\sqrt{1-\mu^2}}\p_{\lambda}+\frac{\cos(\lambda+\Omega t)}{\sqrt{1-\mu^2}} \left((1-\mu^2)\p_{\mu} + \Omega\p_{\psi}\right), \\
     &\JJJ_{3} =\mu \frac{\cos(\lambda+\Omega t)}{\sqrt{1-\mu^2}}\p_{\lambda}-\frac{\sin(\lambda+\Omega t)}{\sqrt{1-\mu^2}} \left((1-\mu^2)\p_{\mu} + \Omega\p_{\psi}\right),
\end{align*}
where the parameter-function $g$ traverses the set of smooth functions of $t$. The structure of the algebra $\mathcal S^\infty_\Omega$ is  $\mathfrak{so}(3)\oplus(\mathfrak{g}_2\lsemioplus \langle\ZZ(g)\rangle)$, where the three-dimensional orthogonal algebra $\mathfrak{so}(3)$ is realized by the vector fields $\JJJ_i$, $i=1,2,3$, $\mathfrak{g}_2 = \langle \DDD, \PP\rangle$ is a realization of the two-dimensional non-Abelian algebra and $\langle\ZZ(g)\rangle$ is an infinite-dimensional Abelian ideal in $\mathcal{S}_\Omega^\infty$.

An important property of the above family of Lie invariance algebras parameterized by the angular velocity~$\Omega$ is that it is not singular with respect to the parameter~$\Omega$ at $\Omega=0$, i.e.\ it includes the case of the rest reference frame, and it is natural to denote the maximal Lie invariance algebra of Eq.~\eqref{eq:VorticityEquationSphere} with $\Omega=0$ by $\mathcal S^\infty_0$.
It was shown in~\cite{bihl09Ay,bihl12Ay} that $\Omega$ can be set to zero in the algebra $\mathcal S^\infty_\Omega$ by means of the transformation
\begin{equation}\label{eq:vorttrans}
    \tilde t = t, \quad \tilde\mu = \mu, \quad \tilde\lambda = \lambda + \Omega t, \quad \tilde\psi = \psi - \Omega\mu.
\end{equation}
The same transformation also allows one to set $\Omega$ to zero in the vorticity equation~\eqref{eq:VorticityEquationSphere}. Note that the transformation~\eqref{eq:vorttrans} was originally derived in~\cite{plat60Ay}, where it was used to transform the vorticity equation into a reference frame with vanishing angular momentum. The transformation~\eqref{eq:vorttrans} can also be found by noting that the algebras $\mathcal S^\infty_\Omega$ and $\mathcal S^\infty_0$ are isomorphic and by constructing the mapping relating these two algebras. For our purpose this transformation is especially convenient as it leads to a simplified form of both the vorticity equation~\eqref{eq:VorticityEquationSphere} and the maximal Lie invariance algebra~$\mathcal S^\infty_\Omega$ (i.e.\ we can work with $\mathcal S^\infty_0$).
Moreover, setting $\Omega$ to zero makes clear the physical meaning of basis elements of the algebra~$\mathcal S^\infty_\Omega$.
If $\Omega=0$, the vector field $\DDD$ generates simultaneous scalings in $t$ and $\psi$ and the vector field $\PP$ is associated with time translations.
The elements of the form $\ZZ(g)$ are the infinitesimal counterparts of gauging of the stream function up to a summand being a smooth function of~$t$.
The vector fields~$\JJJ_i$, $i=1,2,3$, generate rotations represented in angular coordinates.

It was claimed in~\cite{bihl09Ay,bihl12Ay} that in addition to continuous symmetries generated by elements from the maximal Lie invariance algebra~$\mathcal S^\infty_\Omega$ (or, equivalently, $\mathcal S^\infty_0$), there are two discrete symmetries admitted by Eq.~\eqref{eq:VorticityEquationSphere} and these are given by the changes of the signs, $(t,\lambda,\mu,\psi)\mapsto(-t,-\lambda,\mu,\psi)$ and $(t,\lambda,\mu,\psi)\mapsto(t,\lambda,-\mu,-\psi)$, respectively. While it is straightforward to check by direct substitution that these transformations are indeed symmetries of~\eqref{eq:VorticityEquationSphere}, it is more elaborate to derive them directly from the invariance criterion. This was not done in~\cite{bihl09Ay,bihl12Ay}. It is even harder to prove that there are no other independent (up to composition with each other and with continuous symmetry transformations) discrete symmetries than these two mirror symmetries. It is the purpose of this paper to show by determining the complete point symmetry group of the \sbve\ that there are indeed only these two discrete symmetries.

\section{How to find the complete point symmetry group\\ via the algebraic method}\label{sec:DiscreteSymmetries}

It is considerably more difficult to find discrete point symmetries of a system of differential equations than its Lie symmetries. The reason for this is that the powerful infinitesimal symmetry criterion is only applicable for transformations depending on continuous group parameters~\cite{blum89Ay,hydo00Ay,olve86Ay}. This is also the reason why to date no existing computer algebra package, such as~\cite{head93Ay,roch11Ay,vu12Ay} can be used for this purpose, because all such packages rely on the integration of the infinitesimal determining equations, which by definition only exist for Lie symmetries and are linear. Finding discrete point symmetries or the complete point symmetry group of a system of differential equations therefore has in fact to be done by hand using computer programs only for related routine calculations in interactive mode.

Before we start the computation of the complete point symmetry group~$G_0$ for Eq.~\eqref{eq:VorticityEquationSphere} with $\Omega=0$,
let us recall about the general method proposed in~\cite{bihl11Cy},
which in some sense can be seen as a refinement of the technique suggested in~\cite{hydo00By,hydo00Ay} by involving the notion of megaideals~\cite{popo03Ay}.

Namely, we use the following property:
\emph{Given a system of differential equations~$\mathcal L$,
for any transformation~$\mathcal T$ from the maximal point symmetry (pseudo)group~$G$ of the system~$\mathcal L$
the linear mapping~$\mathcal T_*\colon\mathfrak g\to\mathfrak g$ generated by~$\mathcal T$
on the maximal Lie invariance algebra~$\mathfrak g$ of the system~$\mathcal L$
via push-forwarding of vector fields in the space of system variables
is an automorphism of~$\mathfrak g$, $\mathcal T_*\in\mathrm{Aut}(\mathfrak g)$, and hence it preserves all megaideals of~$\mathfrak g$.}

The correspondence $\mathcal T\to\mathcal T_*$ defines a representation of~$G$ on~$\mathfrak g$, which is often unfaithful.
In particular, this is the case if the group~$G$ (resp.\ the algebra~$\mathfrak g$) has a nontrivial center.
The representation image~$G_*$ is a subgroup of the automorphism group $\mathrm{Aut}(\mathfrak g)$, which may be smaller than the entire group $\mathrm{Aut}(\mathfrak g)$.
The continuous point symmetries of the system~$\mathcal L$ generate, via finite compositions, a connected normal subgroup~$U$ of~$G$ and induce mappings on~$\mathfrak g$
which can be considered as internal automorphisms of~$\mathfrak g$
and which generate a normal subgroup of $\mathrm{Aut}(\mathfrak g)$.
We denote this subgroup by $\mathrm{Int}(\mathfrak g)$.
Elements of the factor group $G/U$ (more precisely, their representatives in~$G$)
are interpreted as ``discrete symmetries'' of the system~$\mathcal L$ which are independent up to composing with continuous symmetries of~$\mathcal L$.
As the representations of~$G$ and~$U$ on~$\mathfrak g$ via push-forwarding of vector fields are not necessarily faithful,
there is no assurance on the existence of a bijection between the factor groups $G/U$ and $\mathrm{Aut}(\mathfrak g)/\mathrm{Int}(\mathfrak g)$.
Rigorous consideration gives rise to a number of difficult problems
which concern the relation of algebras of vector fields and (pseudo)groups of transformations in the infinite-dimensional case
and which are beyond the subject of this paper.
At the same time, these problems can be easily solved in particular cases, e.g., related to models of fluid dynamics and meteorology.

If the algebra $\mathfrak g$ is not low dimensional then the computation of $\mathrm{Aut}(\mathfrak g)$ itself may be a complicated problem.
Moreover, the group $\mathrm{Aut}(\mathfrak g)$ may be much wider than~$G_*$, especially if the algebra $\mathfrak g$ is infinite dimensional.
If this is the case, in the course of the construction of $\mathrm{Aut}(\mathfrak g)$ we will spend efforts for finding elements from $\mathrm{Aut}(\mathfrak g)\setminus G_*$,
which are in fact needless for determining $G_*$.
To avoid such needless computations, instead of the condition $G_*\subseteq\mathrm{Aut}(\mathfrak g)$
we can use the weaker condition that $G_*\mathfrak i\subseteq\mathfrak i$ if $\mathfrak i$ is a megaideal of~$\mathfrak g$.

In general, a \emph{megaideal} $\mathfrak i$ of a Lie algebra~$\mathfrak g$ is a vector subspace of $\mathfrak g$ that is invariant under any mapping from the automorphism group $\mathrm{Aut}(\mathfrak g)$ of~$\mathfrak g$ \cite{bihl11Cy,popo03Ay}, i.e., $\mathfrak Tz=z$ for any $z\in\mathfrak i$ and any $\mathfrak T\in\mathrm{Aut}(\mathfrak g)$.
Every megaideal of~$\mathfrak g$ is an ideal and a characteristic ideal of~$\mathfrak g$.
A set of megaideals of~$\mathfrak g$ can be computed without knowing of $\mathrm{Aut}(\mathfrak g)$.
Both the improper subalgebras of~$\mathfrak g$ (the zero subspace and $\mathfrak g$ itself) are (improper) megaideals of~$\mathfrak g$.
If $\mathfrak i_1$ and $\mathfrak i_2$ are megaideals of~$\mathfrak g$ then so are $\mathfrak i_1+\mathfrak i_2,$ $\mathfrak i_1\cap \mathfrak i_2$ and $[\mathfrak i_1,\mathfrak i_2]$,
i.e., sums, intersections and Lie products of megaideals are again megaideals.
If $\mathfrak i_2$ is a megaideal of $\mathfrak i_1$ and $\mathfrak i_1$ is a megaideal of $\mathfrak g$ then $\mathfrak i_2$ is a megaideal of $\mathfrak g$, i.e., megaideals of megaideals are also megaideals.
All elements of the derived series and upper and lower central series of~$\mathfrak g$, including the center and the derivative of~$\mathfrak g$, as well as the radical and nil-radical of~$\mathfrak g$ are its megaideals.
In order to have a sufficient store of megaideals, we need one more way to find new megaideals from known ones.

\begin{proposition}\label{pro:WayToFindMegaideals}
If~$\mathfrak i_0$, $\mathfrak i_1$ and~$\mathfrak i_2$ are megaideals of~$\mathfrak g$
then the set~$\mathfrak s$ of elements from~$\mathfrak i_0$ whose commutators with arbitrary elements from~$\mathfrak i_1$ belong to~$\mathfrak i_2$
is also a megaideal of~$\mathfrak g$.
\end{proposition}

\begin{proof}
It is obvious that $\mathfrak s$ is a linear subspace of~$\mathfrak g$.
Consider an element $z_0\in\mathfrak i_0$ such that $[z_0,z_1]\in\mathfrak i_2$ for arbitrary $z_1\in\mathfrak i_1$.
Then for arbitrary $\mathfrak T\in\mathrm{Aut}(\mathfrak g)$ and arbitrary $z_1\in\mathfrak i_1$ we have
$[\mathfrak Tz_0,z_1]=[\mathfrak Tz_0,\mathfrak T\mathfrak T^{-1}z_1]=\mathfrak T[z_0,\mathfrak T^{-1}z_1]\in\mathfrak i_2$
as $\mathfrak T^{-1}z_1\in\mathfrak i_1$, and hence $[z_0,\mathfrak T^{-1}z_1]\in\mathfrak i_2$.
This means that $\mathfrak Tz_0\in\mathfrak s$, i.e., $\mathfrak s$ is a megaideal of~$\mathfrak g$.
\end{proof}

As the megaideals~$\mathfrak i_1$ and~$\mathfrak i_2$ are necessarily usual ideals and hence $[\mathfrak i_0,\mathfrak i_1]\subseteq\mathfrak i_0\cap\mathfrak i_1$,
it in fact suffices to consider the case when $\mathfrak i_2$ is contained in $\mathfrak i_0\cap\mathfrak i_1$.
If $\mathfrak i_0\cap\mathfrak i_1=\{0\}$, the megaideal~$\mathfrak s$ coincides with~$\mathfrak i_0$.
A particular case of Proposition~\ref{pro:WayToFindMegaideals} with $\mathfrak i_1=\mathfrak g$ and $\mathfrak i_2=\{0\}$ implies that
the centralizer of every megaideal is a megaideal.

Once a convenient set of megaideals is found, one should try to obtain maximal restrictions on the form of a point symmetry transformation~$\mathcal T$, which are feasible using the algebraic approach. If all derived restrictions on $\mathcal T$ are taken into account, one has to substitute the restricted form of a general point transformation into the initial system of differential equations and proceed the computation of the complete point symmetry group within the framework of the direct method.

\section{Computation of the complete point symmetry group}\label{sec:CompletePointSymmetryGroup}

We fix the value $\Omega=0$. It is straightforward to compute the following megaideals of~$\mathfrak g=\mathcal S^\infty_0$:
\begin{gather*}
\mathfrak g'=\langle\JJJ_1,\JJJ_2,\JJJ_3,\PP,\ZZ(g)\rangle,\quad
\mathfrak g''=\langle\JJJ_1,\JJJ_2,\JJJ_3,\ZZ(g)\rangle,\quad
\mathfrak g'''=\langle\JJJ_1,\JJJ_2,\JJJ_3\rangle,\\
\mathrm C_{\mathfrak g}(\mathfrak g')=\mathrm Z_{\mathfrak g'}=\langle\ZZ(1)\rangle,\quad
\mathrm C_{\mathfrak g}(\mathfrak g'')=\mathrm Z_{\mathfrak g''}=\langle\ZZ(g)\rangle,\quad
\mathrm C_{\mathfrak g}(\mathfrak g''')=\langle\DDD,\PP,\ZZ(g)\rangle,\\
(\mathrm C_{\mathfrak g}(\mathfrak g'''))'=\langle\PP,\ZZ(g)\rangle,
\end{gather*}
where $\mathfrak a'$, $\mathrm Z_{\mathfrak a}$ and $\mathrm C_{\mathfrak a}(\mathfrak b)$
denote the derivative and the center of a Lie algebra~$\mathfrak a$
and the centralizer of a subalgebra~$\mathfrak b$ in~$\mathfrak a$, respectively.

Now we apply Proposition~\ref{pro:WayToFindMegaideals} to the case $\mathfrak i_0=\mathfrak i_1=(\mathrm C_{\mathfrak g}(\mathfrak g'''))'$
and vary $\mathfrak i_2$.
If $\mathfrak i_2=\langle\ZZ(1)\rangle$, we obtain $\mathfrak s=\langle\ZZ(1),\ZZ(t)\rangle$ and hence this is a megaideal.
We reassign the last~$\mathfrak s$ as $\mathfrak i_2$ and iterate this procedure,
which results in the series of megaideals $\langle\ZZ(1),\ZZ(t),\dots,\ZZ(t^n)\rangle$, $n\in\mathbb N_0$.

Megaideals of~$\mathcal S^\infty_0$ that are sums of other megaideals are not essential
for the computation of the complete point symmetry group~$G_0$
of the vorticity equation~\eqref{eq:VorticityEquationSphere} for $\Omega=0$ by the algebraic method
since they give weaker constraints for components of point symmetry transformations than their summands.
Even if a megaideal~$\mathfrak i$ is not a sum of other megaideals, 
the condition $G_*\mathfrak i\subseteq\mathfrak i$ may imply only constraints
that are consequences of constraints derived in the course of the consideration of other megaideals.
In order to simplify the computation, we choose a minimal set of megaideals
that allow us to easily obtain a maximal set of constraints available within the algebraic framework.
We selected such megaideals from the above list:
\begin{equation}\label{eq:MegaIdealsForComputation}
\langle\ZZ(1)\rangle,\quad
\langle\ZZ(1),\ZZ(t)\rangle,\quad
\langle\PP,\ZZ(g)\rangle,\quad
\langle\JJJ_1,\JJJ_2,\JJJ_3\rangle.
\end{equation}

The general form of a point transformation that can be applied to the vorticity equation on the sphere~\eqref{eq:VorticityEquationSphere} with $\Omega=0$ is
\[
 \mathcal T\colon\quad (\tilde t, \tilde\lambda,\tilde\mu, \tilde \psi)=(T, \Lambda, \mathrm M, \Psi),
\]
where $T$, $\Lambda$, $\mathrm M$ and $\Psi$ are regarded as functions of $t$, $\lambda$, $\mu$ and $\psi$, whose joint Jacobian~$\mathrm J$ does not vanish.
To derive a constrained form of $\mathcal T$, we use the selected four megaideals~\eqref{eq:MegaIdealsForComputation} of~$\mathcal S^\infty_0$.
For the transformation~$\mathcal T$ to be qualified as a point symmetry of the vorticity equation on the sphere,
its counterpart~$\mathcal T_*$ push-forwarding vector fields should preserve each of these megaideals.
Moreover, for any megaideal~$\mathfrak m$ of~$\mathfrak g$ the mapping induced by~$\mathcal T$ on~$\mathfrak m$ is an automorphism of~$\mathfrak m$.
This property is convenient to use for finite-dimensional megaideals.
Thus, the megaideal $\langle\JJJ_1,\JJJ_2,\JJJ_3\rangle$ is isomorphic to the algebra  $\mathfrak{so}(3)$,
whose automorphism group is exhausted by internal automorphisms and hence isomorphic to the special orthogonal group $\mathrm{SO}(3)$.

As a result, we obtain the conditions
\begin{subequations}
\begin{align}
&\mathcal T_* \ZZ(1)=T_\psi\p_{\tilde t}+\Lambda_\psi\p_{\tilde \lambda}+\mathrm M_\psi \p_{\tilde \mu}+\Psi_\psi \p_{\tilde \psi}=c\tilde\ZZ(1),
\label{eq:MegaidealConstraintForT1}\\
&\mathcal T_* \ZZ(t)=t(T_\psi\p_{\tilde t}+\Lambda_\psi\p_{\tilde \lambda}+\mathrm M_\psi \p_{\tilde \mu}+\Psi_\psi\p_{\tilde\psi})
=d_1\tilde\ZZ(\tilde t)+d_0\tilde\ZZ(1),
\label{eq:MegaidealConstraintForT2}\\
&\mathcal T_* \PP=T_t\p_{\tilde t}+\Lambda_t\p_{\tilde \lambda}+\mathrm M_t \p_{\tilde \mu}+\Psi_t\p_{\tilde \psi}
=a_1\tilde\PP+\tilde\ZZ(\tilde g),
\label{eq:MegaidealConstraintForT3}\\
&\mathcal T_* \JJJ_i=\sum_{j=1}^3 b_{ij}\tilde \JJJ_j,\quad i=1,2,3,
\label{eq:MegaidealConstraintForT4}
\end{align}
\end{subequations}
where $\tilde g$ is a smooth function of~$\tilde t$ that is determined,
as the constant parameters $c$, $d_0$, $d_1$, $a_1$, $a_2$, $a_3$ and~$b_{ij}$,
by~$\mathcal T_*$ and the vector field from the corresponding left-hand side,
$(b_{ij})$ is a special orthogonal matrix, and $i,j=1,2,3$.

We will derive constraints on~$\mathcal T_*$ by sequentially equating the coefficients of vector fields
in the conditions \eqref{eq:MegaidealConstraintForT1}--\eqref{eq:MegaidealConstraintForT4}
and by taking into account the constraints obtained in previous steps.

Thus, condition~\eqref{eq:MegaidealConstraintForT1} directly implies that $T_\psi=\Lambda_\psi=\mathrm M_\psi=0$ and $\Psi_\psi=c$.
Then the last value is nonzero since the Jacobian~$\mathrm J$ does not vanish.
The equation $ct=d_1\tilde t+d_0$ derived from condition~\eqref{eq:MegaidealConstraintForT2}
gives that $d_1\ne0$ and hence the $t$-component of the transformation~$\mathcal T$ depends only on~$t$ and the dependence is affine,
$\tilde t=T(t)=cd_1^{-1}t-d_0d_1^{-1}$.
Condition~\eqref{eq:MegaidealConstraintForT3} is split into the equations $T_t=a_1$ (and hence $a_1=cd_1^{-1}\ne0$), $\Lambda_t=\mathrm M_t=0$ and $\Psi_t=\tilde g$.
Collecting coefficients of~$\p_{\tilde \psi}$ in condition~\eqref{eq:MegaidealConstraintForT4}, we obtain that $\Psi_\lambda=\Psi_\mu=0$.
The integration and arrangement of all the above equations for the components of~$\mathcal T$ results in the representation
\[
T=a_1t+a_0, \quad
\Lambda=\Lambda(\lambda,\mu), \quad
\mathrm M=\mathrm M(\lambda,\mu), \quad
\Psi=c\psi+f(t),
\]
where $a_1$, $a_0$ and $c$ are arbitrary constants with $a_1c\ne0$,
$f$ is an arbitrary smooth function of~$t$,
the pair of the smooth functions~$\Lambda$ and~$\mathrm M$ has nonvanishing Jacobian
and additionally satisfies equations implied by condition~\eqref{eq:MegaidealConstraintForT4}.
Up to internal automorphisms of the algebra~$\mathcal S^\infty_0$
which are generated by the rotation operators~$\JJJ_1$, $\JJJ_2$ and $\JJJ_3$,
we can set the matrix $(b_{ij})$ as the unit matrix.
Then we obtain the following system of equations with respect to the functions~$\Lambda$ and~$\mathrm M$:
\begin{subequations}
\begin{gather}\label{eq:OpsJConditions1}
\JJJ_1\Lambda=1,\quad \JJJ_2\Lambda=\frac{\mathrm M}{\sqrt{1-\mathrm M^2}}\sin\Lambda,\quad \JJJ_3\Lambda=\frac{\mathrm M}{\sqrt{1-\mathrm M^2}}\cos\Lambda,
\\\label{eq:OpsJConditions2}
\JJJ_1\mathrm M=0,\quad \JJJ_2\mathrm M=\sqrt{1-\mathrm M^2}\cos\Lambda,\quad \JJJ_3\mathrm M=-\sqrt{1-\mathrm M^2}\sin\Lambda,
\end{gather}
\end{subequations}

The equations $\JJJ_1\Lambda=1$ and $\JJJ_1\mathrm M=0$ imply that $\Lambda=\lambda+\Upsilon(\mu)$ and $\mathrm M=\mathrm M(\mu)$.
We substitute these expressions into the last two equations of~\eqref{eq:OpsJConditions2}
and split them with respect to~$\lambda$.
This gives the conditions $\sqrt{1-\mathrm M^2}\sin\Upsilon=0$ and $\sqrt{1-\mu^2}\mathrm M_\mu=\sqrt{1-\mathrm M^2}\cos\Upsilon$.
As $\mathrm M_\mu\ne0$, we have that $\sin\Upsilon=0$, i.e.\ $\Upsilon=\pi k$, where $k\in\mathbb Z$.
The same procedure applied to the last two equations of~\eqref{eq:OpsJConditions1} results in
the condition
\[
\frac\mu{\sqrt{1-\mu^2}}=\frac{(-1)^k\mathrm M}{\sqrt{1-\mathrm M^2}},
\]
which is equivalent to $\mathrm M=(-1)^k\mu$.
Then the equation $\sqrt{1-\mu^2}\mathrm M_\mu=\sqrt{1-\mathrm M^2}\cos\Upsilon$ is identically satisfied.

There are no more constraints which can be derived within the framework of the algebraic method.
The further consideration is based on the direct calculation of transformed derivatives,
which is quite easy since the expressions for the transformation components have already been specified.
Thus, the transformed left-hand side of the vorticity equation~\eqref{eq:VorticityEquationSphere} with $\Omega=0$,
\[
\tilde\zeta_{\tilde t}+ (\tilde\psi_{\tilde\lambda}\tilde\zeta_{\tilde\mu}-
\tilde\psi_{\tilde\mu}\tilde\zeta_{\tilde\lambda})=
\frac c{a_1}\zeta_t+(-1)^kc^2(\psi_\lambda\zeta_\mu-\psi_\mu\zeta_\lambda),
\]
identically vanishes for each solution of~\eqref{eq:VorticityEquationSphere}
if and only if $c=(-1)^k/a_1$.
This means that up to rotations, which are generated by vector fields from $\langle\JJJ_1,\JJJ_2,\JJJ_3\rangle$,
any transformation from the group~$G_0$ takes the form
\[
\tilde t=a_1t+a_0, \quad
\tilde\lambda=\lambda, \quad
\tilde\mu=\varepsilon\mu, \quad
\tilde\psi=\frac{\varepsilon}{a_1}\psi+f(t),
\]
where $a_0$ and $a_1$ are arbitrary constants with $a_1\ne0$, $\varepsilon=\pm1$
and $f$ is an arbitrary smooth function of~$t$.
(We neglect the shift of~$\lambda$ by~$\pi k$ as it is a rotation associated with $\JJJ_1$
and denote $(-1)^k$ by~$\varepsilon$.)
A transformation of the above form belongs to the connected component of the unity in~$G_0$ if and only if
$a_1>0$ and $\varepsilon=1$.
Therefore, there are only two discrete transformations in~$G_0$
that are independent up to combinations with each other and with continuous transformations.
These are, e.g., the transformations with $(a_1,\varepsilon)=(-1,1)$ and $(a_1,\varepsilon)=(1,-1)$,
where in both the cases we set $a_0=0$ and $f=0$,
which merely alternate the signs of the variables $\{t,\psi\}$ and $\{\mu,\psi\}$, respectively.

The transformation which alternates the signs of the variables~$\{\lambda,\mu\}$
is in fact not a discrete symmetry of the vorticity equation~\eqref{eq:VorticityEquationSphere} for $\Omega=0$
as it is the rotation by the angle~$\pi$
with respect to the axis corresponding to $\lambda=0$ and $\mu=0$.
The above symmetry transformations alternating signs of different sets of variables can be combined
in order to obtain other pairs of simple discrete transformations
that are independent of each other up to continuous transformations.
An example of such a pair is given by the transformations
merely alternating the signs of the variables $\{t,\lambda\}$ and $\{\mu,\psi\}$, respectively.
These transformations coincide with those stated in~\cite{bihl09Ay,bihl12Ay}.
This completes the description of the complete point symmetry group of the barotropic vorticity equation on the sphere with $\Omega=0$.

By use of the transformation~\eqref{eq:vorttrans} the above discrete transformations can also be transferred to discrete symmetries of the vorticity equation on a constantly rotating sphere.

Summing up the foregoing consideration, we obtain the following assertion.

\begin{theorem}
The complete point symmetry group of the barotropic vorticity equation on the sphere~\eqref{eq:VorticityEquationSphere}
is generated by one-parameter groups associated with vector fields from the algebra~$\mathcal S^\infty_\Omega$ and two discrete transformations, e.g.,
\[
(t,\lambda,\mu,\psi)\mapsto(-t,-\lambda,\mu,\psi)
\quad\mbox{and}\quad
(t,\lambda,\mu,\psi)\mapsto(t,\lambda,-\mu,-\psi).
\]
\end{theorem}

\begin{corollary}
The factor group of the complete point symmetry group of the barotropic vorticity equation on the sphere~\eqref{eq:VorticityEquationSphere}
with respect to its connected component of the unity is isomorphic to the group $\mathbb Z_2\times\mathbb Z_2$.
\end{corollary}

\section{Conclusion}\label{sec:Conclusion}

In this paper we verified the claim raised in~\cite{bihl09Ay,bihl12Ay} that the barotropic vorticity equation on the sphere possesses two independent (up to composition with each other and with continuous symmetry transformations) discrete symmetries. The computation involved two parts, an algebraic step and a step related to the direct method of finding point symmetries. In view of the structure of the maximal Lie invariance algebra $\mathcal S^\infty_0$, we were able to find a sufficiently large number of megaideals of $\mathcal S^\infty_0$ and then selected those of them that were essential for our consideration, i.e.\ the megaideals~\eqref{eq:MegaIdealsForComputation}. This allowed us to derive important restrictions on the form of point symmetry transformations and therefore strongly economized the remaining computations which were necessary to be carried out using the direct method. We should in particular stress that by taking into account all the constraints that are derivable by the algebraic method, we already obtained a strongly restricted form of the admitted point symmetries. Only a single constraint, which relates the constants $a_1$, $\varepsilon$ and~$c$, could not be found from the transformation behavior of the megaideals and consequently had to be determined using the direct method. As the \sbve\ is a complicated third-order nonlinear partial differential equation in $(1+2)$ variables, not deriving the above restricted form would have rendered it quite problematic to compute the complete point symmetry group using only the direct method.

By Proposition~\ref{pro:WayToFindMegaideals} we also extended the number of possibilities to determine megaideals of Lie algebras. This will be crucial for the computation of the complete point symmetry group of other systems of differential equations as the method we proposed in~\cite{bihl11Cy} and applied in this paper heavily relies on the availability of a large number of megaideals of the associated maximal Lie invariance algebras.

Another novel feature of the present paper is the combining of a simplification of automorphisms by factoring out internal automorphisms as originally proposed in~\cite{hydo00By,hydo00Ay} with the algebraic technique based on megaideals. This is advantageous for the case under consideration as the rotations from ${\rm SO}(3)$ in angular coordinates have a rather cumbersome representation, i.e.\ already the direct integration of the Lie equations associated with elements of $\mathfrak{so}(3)$ is a nontrivial problem. If the calculation of the complete point symmetry group~$G_0$ would be done without factoring out internal automorphisms, the integration of the Lie equations would be implicitly repeated during the computation, which would considerably complicate the calculations within the algebraic method.
As $\mathfrak{so}(3)$ is both a direct summand and a megaideal of $\mathcal S^\infty_0$,
the extension of any automorphism of $\mathfrak{so}(3)$ to the complement of $\mathfrak{so}(3)$ in $\mathcal S^\infty_0$ by identity is an automorphism of $\mathcal S^\infty_0$.
Moreover, any such automorphism of $\mathcal S^\infty_0$ is internal as the automorphism group of $\mathfrak{so}(3)$ coincides with the group of internal automorphisms.
Hence we can easily factor out such automorphisms assuming in the course of the computation that the basis elements $\JJJ_1$, $\JJJ_2$ and $\JJJ_3$ are identically transformed.
Factoring out other internal automorphisms does not essentially simplify the consideration.

To conclude, it often happens that some discrete symmetries of a system of differential equations are known but it is difficult to prove that there are no other discrete symmetries. It will therefore be instructive to test the refined algebraic method for the computation of discrete symmetries as presented in this paper with equations which are known to possess nontrivial discrete symmetries, such as the potential fast diffusion equation $v_t=v_{xx}/v_x$, cf.~\cite{popo07Ay}.

\section*{Acknowledgements}

The authors thank Alexander Bihlo for helpful discussions. 
We appreciate the remarks of the anonymous referees and Professor Thomas Witelski which led to improvements of this paper.
This research was supported by the Austrian Science Fund (FWF), projects P20632 and P23714.


\end{document}